\documentclass{article}

\newcommand{\repourl}{\url{https://github.com/felizzi/post-progression-transport-WMHS-NeurIPS-2026}}  

\usepackage[final,dblblindworkshop]{neurips_2026}
\workshoptitle{World Models for High-Stakes Health (WMHS): Reliable Clinical
Trial Simulation and Intervention-Aware Reasoning}

\usepackage[utf8]{inputenc}
\usepackage[T1]{fontenc}
\usepackage{hyperref}
\usepackage{url}
\usepackage{booktabs}
\usepackage{amsmath}
\usepackage{amssymb}
\usepackage{amsthm}
\usepackage{microtype}
\usepackage{graphicx}
\usepackage{multirow}
\usepackage{xcolor}
\usepackage{tikz}
\usetikzlibrary{arrows.meta,positioning,calc}
\usepackage{tikz-cd}

\newtheorem{proposition}{Proposition}
\newtheorem{remark}{Remark}

\newcommand{\rmst}{\mathrm{RMST}}
\newcommand{\SO}{\textsf{S1}}
\newcommand{\ST}{\textsf{S2}}
\newcommand{\SH}{\textsf{S3}}
\newcommand{\tstar}{t^{\star}}

\title{Post-Progression Treatment Regimes and the Limits of\\
Cross-Trial Survival Transport}
\author{%
  Federico Felizzi \\
  SIIAM, Rome, Italy \\
  enGene, Basel, Switzerland \\
  \texttt{federico.felizzi@gmail.com} \\
  \And
  Noman Paracha \\
  Bayer AG, Basel, Switzerland
}

\begin{document}\maketitle

\begin{abstract}
Models that simulate overall survival (OS) from progression-free
survival (PFS) are fitted to published trial reports and used to
transport treatment effects to new populations. That OS depends on what
happens after progression is well known; we show what this implies for
transport. The quantity transport requires depends on the
\emph{post-progression treatment regime}---the rates at which patients
receive subsequent therapy, indexed by biomarker stratum and randomised
arm---and this object is systematically absent from published reports.
Two families of correction are in routine use and neither reaches it.
Crossover adjustment targets the counterfactual in which no patient
receives subsequent therapy, a world no target population inhabits, and
thereby eliminates precisely the information distinguishing source from
target. Population adjustment matches baseline covariates, and a regime
belongs to a trial rather than to a patient, so reweighting individuals
cannot convert one regime into another. We establish the second point
across four adjustment approaches spanning distinct function classes---
moment-matching MAIC, outcome-regression STC, augmented inverse
probability weighting, and a gradient-boosted density ratio---which
carry near-identical bias, showing the failure to be informational
rather than methodological. We exhibit a regime that is exactly balanced
across arms in the marginal, invisible to standard switching assessment,
and nonetheless not transportable. Over plausible target regimes the
effect spans zero even with the mechanism known exactly, so without
knowledge of the target's post-progression landscape the sign of the
treatment effect is undetermined; supplying the target regime to a
regime-aware estimator recovers it. All results
are from simulation and come from a single seeded script with a test
suite, available at \repourl.
\end{abstract}

\section{Introduction}
Cross-trial transport underlies indirect treatment comparison, network
meta-analysis, and synthetic control arms
\citep{signorovitch2010,phillippo2017}, and is increasingly performed
with learned representations of patients or trials
\citep{johansson2016,shalit2017}. A patient world model that generates
OS conditional on PFS is, in effect, an estimator of this transport,
fitted to published trial reports. We ask what such a model can learn
from that evidence base, and identify a specific and consequential gap.

\paragraph{The transport problem.} A randomised trial identifies a
treatment effect in the population it enrolled. Transport asks for the
effect in a different \emph{target} population---another trial's
population, a reimbursement jurisdiction, or a real-world cohort---usually
with individual data from the source trial and only summary data from
the target. Standard methods assume that, conditional on a set of
baseline effect modifiers, the treatment effect is the same in source
and target; differences in the distribution of those modifiers are then
removed by reweighting or regression \citep{colnet2024}. The assumption
fails whenever something that modifies the effect differs between
source and target but is not a baseline covariate.

\paragraph{The core argument.} OS is the sum of PFS and post-progression
survival (PPS). Subsequent therapy is given after progression, so it can
change PPS but not PFS. If the probability of receiving it differs by
arm, by biomarker stratum, or between trials, OS effects differ even when
the drug and the biology are identical. The probability is set by the
treatment landscape a trial runs in, not by any patient characteristic,
so methods that adjust for patient characteristics cannot account for
it.

\paragraph{A motivating pattern.} Consider a trial in which the
progression-free survival hazard ratio is essentially identical in a
biomarker-positive subgroup and its complement, yet the overall survival
hazard ratios diverge sharply---a pronounced OS benefit in the
biomarker-negative stratum and none in the biomarker-positive---while the
marginal rate of subsequent therapy is balanced across arms. The natural
reading is that the biomarker modifies the drug's effect on survival. Our
results show that this reading is not licensed by the data. Because PFS
precedes subsequent therapy it is invariant to the post-progression
regime, whereas OS is not (Figure~\ref{fig:km}); and a regime that is
balanced in the \emph{marginal}---the quantity a standard switching
assessment inspects---can be sharply arm-differential \emph{within}
strata in offsetting directions (regime \ST{},
Section~\ref{sec:regimes}). Holding drug and biology fixed and changing
only the regime, the biomarker-positive OS hazard ratio in our simulation
moves from $0.973$ to $1.235$ to $1.518$ while the PFS hazard ratio does
not move (Table~\ref{tab:published}): the regime alone can produce a
stratum-wise divergence of exactly this kind.
A careful analyst would not stop at the marginal check but would examine
the four arm~$\times$~biomarker cells. That is the right analysis, and it
fails only because published reports do not contain the subsequent-therapy
rate in each of the four cells: they give at most a rate by arm. With the
four cell rates the regime's contribution would be visible at once;
without them, a stratum-specific divergence in OS hazard ratios cannot be
attributed to biology rather than to the within-stratum, by-arm
subsequent-therapy pattern. This paper makes that non-identification
precise.

The quantity transport requires is the treatment effect in the target
population \emph{under the target's own post-progression treatment
regime}. We define a regime as the table $q[m][t]$ giving the
probability that a patient in biomarker stratum $m$ and randomised arm
$t$ receives subsequent targeted therapy after progression. Regimes
differ across trials for reasons unrelated to the drug under study:
approval timelines, reimbursement, geography, and biomarker-based
eligibility for later agents. The gap is not hypothetical. A
cross-sectional review of 275 published oncology trials and 77
registration trials leading to FDA approval found post-progression
treatment assessable in only $36\%$ and $48\%$ of them, and, where
assessable, judged it below the standard of care in $55\%$ and $76\%$
\citep{olivier2023}. The regime is thus often unreported, and often
different from the setting into which results are transported.

Two correction families are in routine use, and they miss the required
quantity in opposite directions. Crossover adjustment---RPSFT
\citep{robins1991}, IPCW \citep{robins2000}, two-stage estimation
\citep{white1999}; see \citet{latimer2014}---reconstructs survival had no
patient received subsequent therapy, and so discards the regime.
Population adjustment---MAIC \citep{signorovitch2010}, STC, and their
representation-learning descendants---reweights or regresses on baseline
covariates, but a regime is a conditional distribution belonging to the
trial, not a baseline covariate.

\paragraph{Contributions.}
\begin{enumerate}\itemsep1pt
\item We exhibit a regime (\ST{}) that is balanced across arms in the
  marginal, and so passes a standard treatment-switching assessment, but
  imbalanced within biomarker strata in offsetting directions. It leaves
  the trial's own estimate nearly unbiased yet is not transportable
  (Sections~\ref{sec:regimes}--\ref{sec:transport}).
\item We show that the failure of population adjustment is
  informational, not methodological: four approaches spanning moment
  matching, outcome regression, doubly robust augmentation, and flexible
  machine learning carry near-identical bias (Section~\ref{sec:four}).
\item We show that crossover adjustment, even with oracle knowledge,
  leaves residual bias larger than the effect being estimated, because
  its estimand is regime-free (Section~\ref{sec:crossover}).
\item We show that over plausible target regimes the effect spans zero
  even with the mechanism known, and that supplying the target regime to
  regime-aware and off-policy estimators recovers it
  (Section~\ref{sec:idset}).
\item We derive a reporting recommendation and an interim remedy:
  publish subsequent therapy by stratum \emph{and} arm, and until then
  report transported OS as a set over plausible target regimes
  (Section~\ref{sec:implications}).
\end{enumerate}

\paragraph{Relation to prior work.} That OS depends on what happens
after progression is well known to trialists and health technology
assessment (HTA) analysts, and that differential subsequent therapy
biases indirect comparison is recognised in applied HTA, where ad hoc
corrections appear \citep{clear2025}. It is not modelled in transport
for a pragmatic reason: the required data are not in published reports.
The general statement that unobserved shifted effect modifiers defeat
baseline adjustment, and that standard of care drifts between trial
eras, is made by \citet{proximal2024} and reviewed by
\citet{colnet2024}. We do not claim to identify the problem. Our
contribution is what follows from it formally: OS transport is not
identified without the regime even when every baseline covariate is
matched (Proposition~\ref{prop:notreweightable}); the failure can be
invisible to the diagnostics practitioners run (\ST{}); no choice of
estimator that ignores it repairs it; the two correction families fail
for complementary reasons; and the regime alone leaves the sign of the
effect undetermined, while knowing it suffices to recover the effect.

\section{Model and estimand}
\label{sec:model}
We work in the accelerated failure time (AFT) framework \citep{wei1992},
in which covariates act multiplicatively on time rather than on the
hazard: a model of the form $\log T = \mu + \beta^\top Z + \varepsilon$,
so that $e^{\beta}$ is a time ratio. AFT is the natural choice here
because OS decomposes additively in time,
$\mathrm{OS}=\mathrm{PFS}+\mathrm{PPS}$, which a time-scale model
represents directly and a proportional hazards model does not. We
therefore model PFS and PPS separately and add them, rather than
regressing OS on PFS. For patient $i$ with randomised treatment
$T_i\sim\mathrm{Bern}(1/2)$, biomarker $M_i\sim\mathrm{Bern}(\pi)$, and
baseline covariates $X_i\in\mathbb R^5$:
\begin{align}
\log S_i &= \log S_0 + \delta T_i + b^\top (X_i-\bar X)
   + \varepsilon^S_i, \label{eq:pfs}\\
A_i &\sim \mathrm{Bern}\big(q[M_i][T_i]\big), \label{eq:regime}\\
\log \mathrm{PPS}_i &= \kappa + \lambda_{M_i}\log(S_i/S_0)
   + \theta_{T,M_i}T_i + \theta_A A_i + c^\top (X_i-\bar X)
   + \varepsilon^P_i, \label{eq:pps}\\
Y_i &= S_i + \mathrm{PPS}_i . \label{eq:os}
\end{align}
Here $S$ denotes PFS, $Y$ denotes OS, and $A$ denotes subsequent
targeted therapy. Equation~\eqref{eq:pfs} is an AFT model for PFS: $S_0$
is the control-arm median PFS at the source covariate mean $\bar X$
($6$ months), $\delta$ the treatment's log time ratio on PFS, and $b$
the covariate effects. Equation~\eqref{eq:regime} draws subsequent
therapy from the regime. Equation~\eqref{eq:pps} is an AFT model for
PPS: $\kappa$ sets its baseline scale, $\theta_A>0$ is the log time
ratio of subsequent therapy on PPS, and $c$ the covariate effects. The
errors are independent Gaussian, $\varepsilon^S_i\sim N(0,0.45^2)$ and
$\varepsilon^P_i\sim N(0,0.55^2)$, so PFS and PPS are log-normal; all
parameter values are in the released code.

Two parameters carry the mechanism linking the endpoints:
$\lambda_m$, the elasticity of post-progression survival to PFS, where
$\lambda=0$ means a PFS gain passes through to OS one-for-one and
$\lambda=-1$ means it is fully absorbed by shorter survival after
progression; and $\theta_{T,m}$, a direct treatment effect on PPS
bypassing the surrogate. Only the combination
$\lambda_m\delta+\theta_{T,m}$ is visible at a fixed treatment effect
size, a point we return to in Section~\ref{sec:idset}. Because $\delta$
is a log time ratio, it is not a log hazard ratio; the two coincide only
for Weibull or exponential survival.

Administrative censoring follows uniform accrual and a fixed analysis
date. Because $A$ enters only \eqref{eq:pps}, PFS is invariant to the
regime by construction, a fact we exploit throughout.

\paragraph{Notation for the transport problem.} Write $\mathcal S$ for
the source trial and $\mathcal T$ for the target population. A study is
characterised by the triple
\begin{equation}
\big(\,P(X),\ \mathcal M,\ q\,\big),
\qquad
\mathcal M = \{\delta,\ \lambda_m,\ \theta_{T,m},\ \theta_A,\ \kappa\},
\label{eq:triple}
\end{equation}
where $P(X)$ is the baseline covariate distribution including biomarker
prevalence $\pi$, $\mathcal M$ collects the structural parameters of
\eqref{eq:pfs}--\eqref{eq:os}, and $q$ is the post-progression regime of
\eqref{eq:regime}. Cross-trial transport assumes $\mathcal M$ is shared
and adjusts for the difference in $P(X)$. The premise of this paper is
that $q$ also differs, and that it is neither shared nor adjusted for.

\paragraph{Estimand.} We report differences in restricted mean survival
time at horizon $\tstar=36$ months, estimated by weighted Kaplan--Meier
with pooled tie handling. RMST is used rather than the hazard ratio
because it is collapsible \citep{royston2013,uno2014}: a marginal hazard
ratio is not a weighted average of stratum-specific hazard ratios even
absent confounding, which breaks the population reweighting transport
requires. The target estimand is
\begin{equation}
\tau_{\mathcal T}
 = \rmst_1(\tstar) - \rmst_0(\tstar)
 \quad\text{in population } (\pi_{\mathcal T}, X_{\mathcal T}),
 \text{ under regime } q_{\mathcal T},
\label{eq:estimand}
\end{equation}
where $\rmst_t(\tstar)$ is the mean OS truncated at $\tstar$ in arm $t$.
The final clause is what this paper concerns. Hazard ratios are used
only descriptively, to show what a trial would report
(Section~\ref{sec:regimes}).

\section{Three regimes, and what a trial would report}
\label{sec:regimes}
Table~\ref{tab:regimes} defines three regimes. In \SO{} the arms are
balanced within each stratum, with biomarker-positive patients receiving
more subsequent therapy---realistic, since a biomarker often confers
eligibility for a targeted agent. In \SH{} the control arm receives more
in both strata: textbook crossover, and what
\citet{robins1991,white1999} were designed to repair. \ST{} is the case
of interest: the arms differ within strata by $+0.15$ and $-0.28$, in
opposite directions, tuned so the marginal rates coincide at the source
prevalence.

\begin{table}[t]\centering
\caption{Post-progression regimes, given as treatment arm / control arm.
A standard switching assessment inspects the marginal difference;
observed marginals are from $n=240{,}000$ at prevalence $0.35$.}
\label{tab:regimes}
\begin{tabular}{lccrc}
\toprule
regime & biomarker $-$ & biomarker $+$ & marginal diff. & flagged? \\
\midrule
\SO{} balanced   & 0.20 / 0.20 & 0.60 / 0.60 & $+0.0005$ & no \\
\ST{} offsetting & 0.40 / 0.25 & 0.40 / 0.68 & $-0.0012$ & no \\
\SH{} classic    & 0.10 / 0.50 & 0.30 / 0.80 & $-0.4345$ & yes \\
\bottomrule
\end{tabular}
\end{table}

Table~\ref{tab:published} shows what a publication would report for each.
PFS hazard ratios are identical across regimes, as they must be:
subsequent therapy cannot affect an endpoint that precedes it. OS
hazard ratios are not. In the biomarker-positive stratum they run
$0.973$, $1.235$, $1.518$---from no apparent effect to apparent harm---
with identical biology, identical drug, and identical PFS results.
Figure~\ref{fig:km} shows the same fact as survival curves.

\begin{table}[t]\centering
\caption{Simulated trials as a publication would report them
($n=240{,}000$, prevalence $0.35$). Biology and drug are identical across
the three panels; only the post-progression regime differs.}
\label{tab:published}
\begin{tabular}{llccc}
\toprule
regime & subgroup & PFS HR & OS HR & OS events \\
\midrule
\multirow{3}{*}{\SO{} balanced}
 & overall        & 0.630 & 0.575 & 38.0\% \\
 & biomarker $-$  & 0.633 & 0.437 & 38.4\% \\
 & biomarker $+$  & 0.624 & 0.973 & 37.4\% \\
\midrule
\multirow{3}{*}{\ST{} offsetting}
 & overall        & 0.630 & 0.582 & 36.9\% \\
 & biomarker $-$  & 0.633 & 0.384 & 36.1\% \\
 & biomarker $+$  & 0.624 & \textbf{1.235} & 38.3\% \\
\midrule
\multirow{3}{*}{\SH{} classic}
 & overall        & 0.630 & 0.841 & 36.8\% \\
 & biomarker $-$  & 0.633 & 0.616 & 36.1\% \\
 & biomarker $+$  & 0.624 & \textbf{1.518} & 38.1\% \\
\bottomrule
\end{tabular}
\end{table}

\begin{figure}[t]\centering
\includegraphics[width=\linewidth]{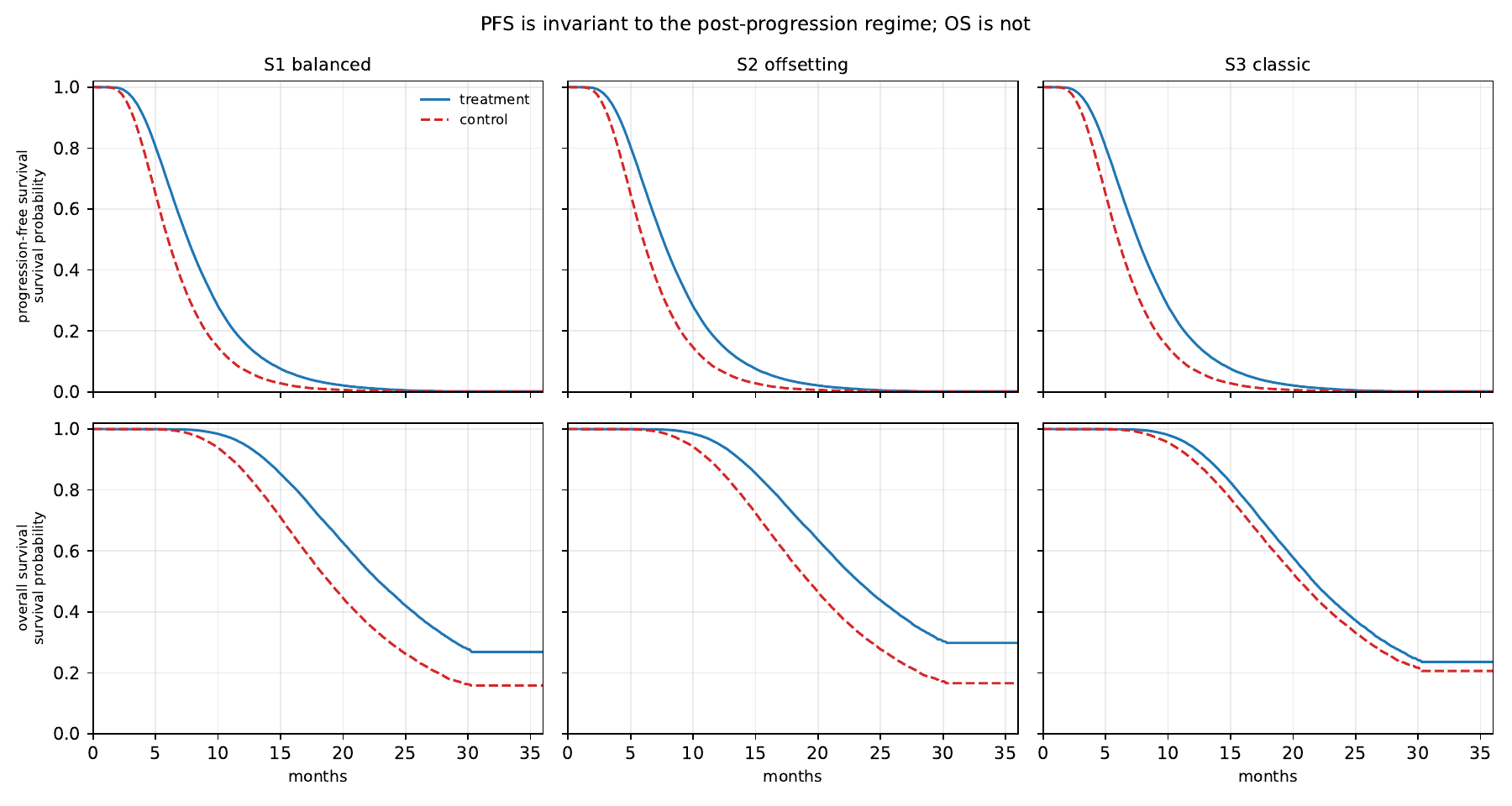}
\caption{Kaplan--Meier estimates by arm for each regime ($n=120{,}000$
per trial; at this size the step function appears smooth). The PFS
panels are identical across regimes; the OS panels are not. Subsequent
therapy cannot affect an endpoint that precedes it, but it reshapes the
endpoint that follows. Each curve is flat beyond $30.4$ months, the
maximum follow-up under administrative censoring.}
\label{fig:km}
\end{figure}

\begin{remark}[Apparent surrogacy modification is manufacturable]
\label{rem:surrogacy}
The stratum-specific translation of a PFS effect into an OS effect---the
object of surrogate validation \citep{prentice1989,buyse2000}, and of
claims that a biomarker modifies surrogacy---varies across
Table~\ref{tab:published} with no biological difference whatsoever.
Meta-analytic surrogate validation pools trials across eras with
differing regimes, so the resulting association mixes biology with
policy.
\end{remark}

Within the trial, \ST{} is nearly as harmless as the balanced \SO{}: its
offsetting imbalance cancels at the prevalence where it was constructed,
leaving a bias of $-0.111$ months relative to the no-subsequent-therapy
effect, against $-0.054$ for \SO{} and $-2.693$ for \SH{}.
Section~\ref{sec:transport} shows that it is nonetheless not
transportable, because that cancellation is prevalence-specific.

\section{Transport under population adjustment}
\label{sec:transport}
We simulate a source trial (prevalence $0.35$) and a target population
(prevalence $0.70$, shifted baseline covariates), pair them across
regimes, and apply matching-adjusted indirect comparison (MAIC)
\citep{signorovitch2010}. MAIC reweights source patients so that their
weighted covariate means equal the published target means:
$w_i=\exp\{a^\top(X_i-\mu_{\mathcal T})\}$, with $a$ solving the
moment-matching equations. It assumes that every effect modifier is
among the matched covariates. Five baseline covariates plus biomarker
prevalence are matched, achieving an effective sample size (ESS) of
$49.3\%$. Table~\ref{tab:transport} reports bias over 20 independent
replications.

\begin{table}[t]\centering
\caption{Transport bias (months of RMST), mean over 20 replications at
$n=60{,}000$. Baseline covariates and biomarker prevalence are matched
exactly in every row.}
\label{tab:transport}
\begin{tabular}{llrrr}
\toprule
source & target & mean bias & sd & mean ESS \\
\midrule
\SO{} & \SO{} & $+0.092$ & $0.143$ & 49.3\% \\
\ST{} & \ST{} & $+0.022$ & $0.168$ & 49.3\% \\
\SO{} & \ST{} & $\mathbf{+0.949}$ & $0.132$ & 49.3\% \\
\ST{} & \SO{} & $\mathbf{-0.835}$ & $0.179$ & 49.3\% \\
\SH{} & \SO{} & $-2.748$ & $0.169$ & 49.3\% \\
\bottomrule
\end{tabular}
\end{table}

Matched-regime bias is centred near zero, with a replication standard
deviation of about $0.15$ months; that standard deviation is the yardstick
against which the mismatched rows should be read, and they sit five to six
standard deviations away.

\begin{proposition}[Regimes are not reweightable]
\label{prop:notreweightable}
Let $Z_i=(X_i,M_i,T_i,S_i,Y_i,A_i)$ together with censoring indicators
denote the observed source data for patient $i$, and let
$\hat\tau_w$ be any estimator formed by reweighting source individuals,
$\hat\tau_w=\sum_i w_i\,g(Z_i)\big/\sum_i w_i$, for an arbitrary
function $g$ and weights $w_i=w(X_i)$ depending on the covariates only.
Then
\begin{equation}
\hat\tau_w \;=\; \Phi\big(P_{\mathcal S}(Z),\ X_{\mathcal T}\big),
\qquad
\frac{\partial \hat\tau_w}{\partial q_{\mathcal T}} \;=\; 0 ,
\label{eq:noreweight}
\end{equation}
whereas the estimand \eqref{eq:estimand} generically satisfies
$\partial\tau_{\mathcal T}/\partial q_{\mathcal T}\neq 0$ when
$\theta_A\neq 0$. Hence, across targets that differ only in
$q_{\mathcal T}$ and have different $\tau_{\mathcal T}$, $\hat\tau_w$ is
consistent for at most one.
\end{proposition}

\begin{proof}[Proof sketch]
By the law of large numbers
$\hat\tau_w\to\mathbb E_{\mathcal S}[w(X)g(Z)]/\mathbb E_{\mathcal S}[w(X)]$.
The expectation is over the source law $P_{\mathcal S}(Z)$, generated
under $q_{\mathcal S}$, and $w$ is fixed by $P_{\mathcal S}(X)$ and the
target covariates; neither involves $q_{\mathcal T}$. The estimand is the
RMST difference of $Y=S+\mathrm{PPS}$ generated under
\eqref{eq:regime} with $q=q_{\mathcal T}$. Raising $q_{\mathcal T}[m][t]$
moves mass in arm $t$ to $A=1$, which multiplies PPS by $e^{\theta_A}$
and so changes that arm's RMST; a change that differs between arms
therefore moves $\tau_{\mathcal T}$, except on a measure-zero set of
parameter values.
\end{proof}

\noindent The content is that a regime is a conditional distribution
$q[m][t]$ belonging to a trial, not a covariate value belonging to a
patient. Reweighting changes the distribution of covariates; it cannot
change the policy under which patients were treated. A patient with
$A=1$ under \SO{} is not exchangeable with a patient with $A=1$ under
\ST{}, because what determined their subsequent treatment differed.

Adding $A$ to the matching set is no remedy. It conditions on a
post-randomisation mediator, which alters the estimand and forfeits the
protection randomisation provides, and it still does not supply
$q_{\mathcal T}$.

\section{Four adjustment approaches, one failure}
\label{sec:four}
If the obstruction were the parametric form of MAIC's weights, a more
flexible estimator would help. We compared four approaches spanning
distinct function classes. \emph{MAIC} reweights by moment matching, as
above. \emph{Simulated treatment comparison} (STC) fits an outcome
regression in the source and evaluates it at the target covariate
profile; it assumes the regression captures all effect modification.
\emph{Augmented inverse probability weighting} (AIPW) combines the two
and is doubly robust. The \emph{ML density ratio} trains a
gradient-boosted classifier to distinguish source from target
individuals and uses its odds as weights, removing MAIC's log-linear
restriction on the weights. The AIPW estimator is given in
Appendix~\ref{app:aipw}.

\begin{table}[t]\centering
\caption{Bias of four adjustment approaches (months of RMST), mean over
10 replications with standard deviation in parentheses. All four are
functions of source data and target \emph{covariates}; none is a function
of the target regime.}
\label{tab:four}
\begin{tabular}{llrrrr}
\toprule
source & target & MAIC & STC & AIPW & ML ratio \\
\midrule
\SO{} & \SO{} & $+0.115\,(0.12)$ & $-0.223\,(0.12)$ & $-0.008\,(0.15)$ & $+0.162\,(0.13)$ \\
\ST{} & \ST{} & $+0.052\,(0.14)$ & $+0.170\,(0.11)$ & $-0.019\,(0.15)$ & $+0.116\,(0.14)$ \\
\SO{} & \ST{} & $\mathbf{+0.950}\,(0.12)$ & $\mathbf{+0.612}\,(0.12)$ & $\mathbf{+0.828}\,(0.14)$ & $\mathbf{+0.997}\,(0.12)$ \\
\ST{} & \SO{} & $\mathbf{-0.783}\,(0.14)$ & $\mathbf{-0.665}\,(0.11)$ & $\mathbf{-0.855}\,(0.16)$ & $\mathbf{-0.719}\,(0.13)$ \\
\bottomrule
\end{tabular}
\end{table}

All four fail together (Table~\ref{tab:four}). The obstruction is
informational: the regime is not a function of the data these methods
consume, so no estimator operating on that data can recover it, whatever
its capacity.

Double robustness makes the point precise, and sharpens it. AIPW
\eqref{eq:aipw} is consistent if \emph{either} its outcome model $\hat m$
or its weights $w$ are correctly specified. But both are maps out of the covariate space,
\begin{equation}
\hat m:\ \mathcal X\times\{0,1\}\to\mathbb R,
\qquad
w:\ \mathcal X\to\mathbb R_{>0},
\qquad\text{so}\qquad
\partial\hat\tau_{\text{AIPW}}/\partial q_{\mathcal T}=0,
\label{eq:drfails}
\end{equation}
and the protection double robustness affords is orthogonal to the
failure at issue. Formally, writing $\mathcal F$ for misspecification of
functional form and $\mathcal I$ for absence of information,
double robustness gives consistency under $\mathcal F$ and is silent
under $\mathcal I$; the regime shift is an instance of $\mathcal I$.

Appendix~\ref{app:aipw} verifies this numerically: under deliberate
misspecification with matched regimes AIPW bias stays within
$\pm0.08$, while with mismatched regimes and both components correct it
is an order of magnitude larger.

\paragraph{A representation-learning reading.} Standard
non-identifiability in causal representation learning
\citep{scholkopf2021} concerns information present in the observations
but not separable into latent factors. Here it is absent: a sufficient
statistic for \eqref{eq:estimand} includes $q[m][t]$, which trial reports
do not contain, so no architecture or scale can recover it;
Table~\ref{tab:four} is its empirical shadow.

\section{Crossover adjustment targets a regime-free world}
\label{sec:crossover}
Crossover adjustment methods estimate survival as it would have been had
no patient received subsequent therapy. RPSFT \citep{robins1991} assumes
subsequent therapy rescales remaining survival time by a common factor
and estimates that factor by g-estimation; IPCW \citep{robins2000}
censors patients when they switch and reweights those who remain; the
two-stage method \citep{white1999} estimates the effect of switching
among progressors and removes it from their post-progression time. We
implement a two-stage-style estimator. Post-progression survival is
recoverable as $\mathrm{PPS}=Y-S$, and since $A$ enters \eqref{eq:pps}
multiplicatively, $\mathrm{PPS}^{\mathrm{adj}}
=\mathrm{PPS}\,e^{-\theta_A A}$ and $Y^{\mathrm{adj}}=S+
\mathrm{PPS}^{\mathrm{adj}}$, targeting
$\mathbb E[Y\mid do(T{=}t), do(A{=}0)]$.

On its own terms the method works: within trial it recovers the
no-therapy counterfactual in every regime, repairing \SH{}'s $-2.64$
month bias to within $0.05$ (Table~\ref{tab:crossover}, upper). Our
implementation reconstructs PPS from uncensored quantities and is thus a
best case; a real analyst has less.

\begin{table}[t]\centering
\caption{Crossover adjustment. Upper: within trial it recovers the
no-therapy counterfactual. Lower: transported, every regime adjusts to
the same value, because that counterfactual is regime-free.}
\label{tab:crossover}
\begin{tabular}{lrrrr}
\toprule
regime & unadjusted & adjusted & oracle & $A{=}0$ truth \\
\midrule
\SO{} balanced   & $+3.705$ & $+3.821$ & $+3.818$ & $+3.775$ \\
\ST{} offsetting & $+3.702$ & $+3.825$ & $+3.818$ & $+3.775$ \\
\SH{} classic    & $+1.132$ & $+3.783$ & $+3.818$ & $+3.775$ \\
\midrule
source $\to$ target & \multicolumn{2}{r}{oracle-adjusted}
  & target truth & bias \\
\midrule
\SH{} $\to$ \SO{} & \multicolumn{2}{r}{$+3.818$} & $+1.520$ & $+2.298$ \\
\SO{} $\to$ \ST{} & \multicolumn{2}{r}{$+3.818$} & $+0.668$ & $+3.150$ \\
\ST{} $\to$ \SO{} & \multicolumn{2}{r}{$+3.818$} & $+1.520$ & $+2.298$ \\
\bottomrule
\end{tabular}
\end{table}

Transported, it answers a different question from the one asked. All
three regimes adjust to $+3.818$, against target truths of $+0.668$ and
$+1.520$: residual bias of $2.3$ to $3.2$ months, larger than the effects
themselves. The method does not fail to see the regime; it removes it,
and the regime is what distinguishes source from target. Adjustment also
collapses \SO{}, \ST{}, and \SH{} to a common value, destroying the
information that would signal whether transport is safe.

\section{The identified set, and what the regime buys}
\label{sec:idset}
If the regime is unknown, the honest output is a set rather than a point
\citep{manski2003,imbens2004}. Let $\mathcal G(B)$ denote the mechanisms
and $\mathcal Q(B)$ the regimes consistent with background knowledge $B$
and the published data. The identified set and its width are
\begin{equation}
\Psi(B)\;=\;\Big\{\,\tau_{\mathcal T}(\mathcal M, q)\ :\
  \mathcal M\in\mathcal G(B),\ q\in\mathcal Q(B)\,\Big\},
\qquad
W(B)\;=\;\sup\Psi(B)-\inf\Psi(B).
\label{eq:idset}
\end{equation}
Two indeterminacies enter. Mechanisms: only $\lambda_m\delta+\theta_{T,m}$
is visible at fixed $\delta$, so a one-parameter family of
$(\lambda,\theta_T)$ reproduces any observed stratum-specific hazard
ratio. Regimes: the target may plausibly operate under \SO{}, \ST{}, or
\SH{}. We quantify the second alone, holding the mechanism at its true
value. Since that is a subset of $\Psi(B)$, the width we report is a lower
bound on $W(B)$.

\paragraph{The set spans zero.} With the mechanism known exactly, the
true target effect is $+1.49$ months under \SO{}, $+0.66$ under \ST{} and
$-1.26$ under \SH{} (means of 6 replications), so
$\Psi\supseteq[-1.26,+1.49]$, of width $2.75$. Without knowledge of the
target's post-progression landscape the \emph{sign} of the treatment
effect in the target population is undetermined, before any mechanistic uncertainty
is added.

\paragraph{Knowing the regime is sufficient.} The converse also holds in
our simulation. Treating the regime as a policy, a regime-aware estimator
supplied with $q_{\mathcal T}[m][t]$ (g-computation \citep{robins1986} over
the target policy) recovers the target effect
across all five source--target pairs of Table~\ref{tab:transport}, with
mean bias between $-0.05$ and $+0.03$ months, against MAIC biases of
$+0.90$, $-0.87$ and $-2.74$ on the mismatched pairs in the same
replications. Three off-policy estimators---the direct method,
importance weighting on the policy ratio $q_{\mathcal T}/q_{\mathcal S}$,
and their doubly robust combination---give mean biases between $-0.12$ and
$+0.05$, each within about one standard error of zero (6 replications).
They assume a PPS mechanism shared between source and target, and
overlap: the source must treat every stratum--arm cell the target policy
reaches, as it does here.

\section{Conclusions}
\label{sec:implications}
The results reduce to a single statement about what the evidence base
contains. Transport requires $\tau_{\mathcal T}$ as defined in
\eqref{eq:estimand}, a functional of the triple
$(P_{\mathcal T}(X), \mathcal M, q_{\mathcal T})$. Published reports
determine the first argument, constrain the second only up to the
one-parameter family of Section~\ref{sec:idset}, and are silent on the
third. The two correction families in routine use fail on the third
argument in complementary ways:
\begin{equation}
\underbrace{\hat\tau_{\text{adj}}
  \;=\;\tau\big(P_{\mathcal S}(X),\mathcal M, q\equiv 0\big)}
  _{\text{crossover adjustment: sets } q \text{ to zero}}
\qquad\qquad
\underbrace{\hat\tau_{w}
  \;=\;\tau\big(P_{\mathcal T}(X),\mathcal M, q_{\mathcal S}\big)}
  _{\text{population adjustment: retains } q_{\mathcal S}}
\label{eq:twofailures}
\end{equation}
Neither is $\tau(P_{\mathcal T}(X),\mathcal M,q_{\mathcal T})$. One
erases the regime; the other carries the wrong one. Figure~\ref{fig:square}
places the four quantities in the plane of the triple \eqref{eq:triple}:
each correction moves along a single axis and neither reaches the
target's own regime.

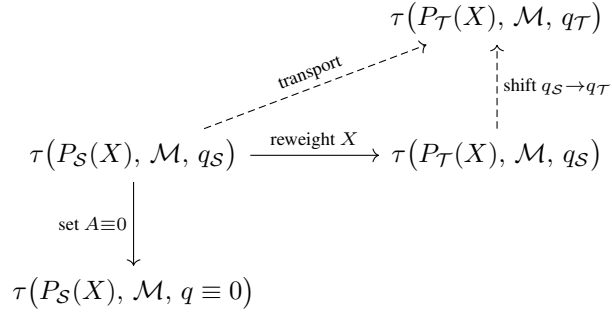
\begin{figure}[t]\centering
\begin{tikzcd}[column sep=4.2em, row sep=3.2em, ampersand replacement=\&]
  \& \tau\big(P_{\mathcal T}(X),\,\mathcal M,\,q_{\mathcal T}\big)
     \arrow[d, phantom, ""{coordinate, name=T}]
  \\
  \tau\big(P_{\mathcal S}(X),\,\mathcal M,\,q_{\mathcal S}\big)
     \arrow[r, "\text{reweight } X"]
     \arrow[d, "\text{set } A\equiv 0"']
     \arrow[ur, dashed, "\text{transport}"{sloped, above}]
  \& \tau\big(P_{\mathcal T}(X),\,\mathcal M,\,q_{\mathcal S}\big)
     \arrow[u, dashed, "\text{shift } q_{\mathcal S}\to q_{\mathcal T}"']
  \\
  \tau\big(P_{\mathcal S}(X),\,\mathcal M,\,q\equiv 0\big)
  \&
\end{tikzcd}
\caption{The transport problem as a diagram in the triple
\eqref{eq:triple}. Solid arrows are the moves the two correction
families make: population adjustment changes the first argument and
lands at $\hat\tau_w$; crossover adjustment changes the third and lands
at $\hat\tau_{\text{adj}}$. Dashed arrows are the moves transport
requires. By Proposition~\ref{prop:notreweightable} no function of
source data and target covariates can perform them.}
\label{fig:square}
\end{figure}

\noindent Four consequences follow.

\paragraph{For outcome simulators.} A model generating OS from PFS must
represent $q$ explicitly. A model class unable to express $q[m][t]$
cannot represent the phenomenon at any capacity, and a model fitted to
trial reports selects one point of $\Psi(B)$ while reporting it as the
answer. The honest output is the set, or the point together with the
assumption on $q$ that identifies it.

\paragraph{For validation.} Validating a simulator against published
transported estimates cannot distinguish a model with the regime right
from one with it wrong, since by \eqref{eq:noreweight} every estimator in
Table~\ref{tab:four} reproduces the published quantities while differing
on the target. Validation must target quantities that separate regimes:
post-progression survival and subsequent-therapy rates by stratum and
arm. Where target outcomes are available, a simple diagnostic follows
from regime invariance of PFS: if transported PFS agrees with the target
but transported OS does not, the discrepancy lies on the post-progression
pathway, and the regime is the first suspect.

\paragraph{For analysis now.} Until the regime is reported, a
transported OS estimate should be accompanied by the identified set
\eqref{eq:idset} computed over a support $\mathcal Q(B)$ for the target
regime. That support can be elicited from the target's treatment
landscape---which later-line agents are approved and reimbursed there,
and for which biomarker strata---information that is typically
available even when the trial-level regime is not. Where the resulting
set excludes zero, transport supports a decision despite the missing
regime; where it spans zero, as in our example, the analysis has
identified what must be learned before it can; where $q_{\mathcal T}$ is
known, Section~\ref{sec:idset}'s estimators apply.

\paragraph{For trial reporting.} By Section~\ref{sec:idset} the regime
alone moves the target effect across zero, and knowing it suffices to
recover the effect, so the highest-value change is publishing subsequent
therapy by biomarker stratum \emph{and} arm. Marginal rates are insufficient:
Table~\ref{tab:regimes} exhibits two regimes with indistinguishable
marginals, $+0.0005$ and $-0.0012$, and materially different transport
behaviour. Two clarifications. First, transport needs the regime as a
\emph{policy}, the rate table $q[m][t]$, not each patient's completed
history; it is estimable at data cutoff from progressed patients, with
death before subsequent therapy as a competing risk, without waiting for
OS. Second, the current baseline is low:
post-progression treatment is assessable in under half of registration
trials and about a third of published trials \citep{olivier2023}, and
reporting by stratum and arm is rarer still. Reporting it so raises real
difficulties---small cells, harmonising therapy classes across
jurisdictions, and data-sharing and privacy constraints---of the same
kind as those already managed for subgroup efficacy reporting.

\section{Limitations and responsible use}
All results are from simulations designed to exhibit the mechanism.
They establish that the failure can occur, why no covariate adjustment
repairs it, and that knowing the regime does; they do not establish how often regimes differ across real trials or how
large the bias is in any real comparison. Measuring that requires
subsequent therapy reported by stratum and arm, the very gap identified
here; calibrating the simulation to published by-arm rates is a natural
next step. Parameters are illustrative, indicating scale, not size in any
specific comparison. The weighted
Kaplan--Meier extrapolates flat beyond the last observed time ($30.4$
months), overstating the RMST difference by about $0.16$ months against a
direct Monte Carlo integral; the bias is common-mode and largely cancels
in bias columns. Covariate shifts are moderate, because aggressive shifts
make MAIC fail on its own terms through a low effective sample size. The
crossover adjustment, the STC and AIPW outcome models, and the
regime-aware and off-policy estimators use uncensored outcomes and are
best cases. The set in Section~\ref{sec:idset} holds the mechanism fixed
and so understates the full identified set; it also depends on the
chosen support of regimes.

\paragraph{Responsible use.} This paper presents methodological
contributions to causal inference for cross-trial comparison. All data
are simulated. Nothing here has direct clinical implications, and no
result should inform treatment decisions, clinical practice, drug
approval, or reimbursement policy. No specific trial is named anywhere
in this work; the regimes and parameters are stylized, and any
resemblance to a particular comparison is incidental. The anticipated
misuse is invoking the non-identification results as grounds to reject
cross-trial evidence in general. Sections~\ref{sec:idset} and
\ref{sec:implications} give the mitigation: once the target regime is
known, regime-aware transport recovers the target effect, and the
reporting this requires is feasible.

\paragraph{Reproducibility.} All results are produced by a single seeded
script with a test suite; tabulated results are written to
machine-readable files and the remainder to a console log. The code,
outputs and log are available at \repourl.

\paragraph{Use of large language models.} A large language model was used
to assist with drafting the manuscript, preparing figures, and checking
the reference list. All modelling choices, derivations, code, and reported
numbers are the authors' own; the authors verified them and take full
responsibility for the content.

\bibliographystyle{plainnat}
\bibliography{references}

\appendix
\section{The AIPW estimator and its double robustness}
\label{app:aipw}
The AIPW estimator of Section~\ref{sec:four}, with truncated outcome
$\tilde Y_i=\min(Y_i,\tstar)$, outcome model $m(x,t)$, weights $w_i$,
and randomisation probability $e=1/2$, is
\begin{equation}
\hat\tau_{\text{AIPW}}
= \frac{1}{n_{\mathcal T}}\sum_{j\in\mathcal T}
   \big[\hat m(X_j,1)-\hat m(X_j,0)\big]
 + \frac{1}{n_{\mathcal S}}\sum_{i\in\mathcal S} w_i
   \left(\frac{T_i}{e}-\frac{1-T_i}{1-e}\right)
   \big(\tilde Y_i-\hat m(X_i,T_i)\big).
\label{eq:aipw}
\end{equation}
The first term is STC; the second is a weighted average of source
residuals. If $\hat m$ is correct the residuals have mean zero and the
correction vanishes in expectation; if $w$ is correct the correction
supplies exactly the bias a wrong $\hat m$ introduces. Either suffices,
which is the double robustness property.

\paragraph{Numerical check.} With regimes \emph{matched} and
misspecification deliberately induced, AIPW bias is $-0.077$ with both
components correct, $-0.050$ with the outcome model degraded, $+0.072$
with the weights degraded, and $+0.057$ with both degraded---all within
$\pm 0.08$, confirming the property holds. With regimes
\emph{mismatched} and both components correctly specified, bias is
$+0.828$ and $-0.855$ (Table~\ref{tab:four}): an order of magnitude
larger, in exactly the setting where double robustness offers nothing.
\end{document}